\documentclass[conference]{IEEEtran}
\twocolumn
\usepackage[english]{babel}
\usepackage{amsmath,amsthm}
\usepackage{amsfonts}
\usepackage[final,hyperfootnotes=false,hidelinks]{hyperref}
\usepackage[final]{graphicx}
\usepackage{epstopdf}
\usepackage{tikz}
\usepackage{pgfplots}
\usepackage{subfig}
\usepackage{flushend}
\usepackage{cprotect}

\usepackage{tikz}
\usetikzlibrary{arrows,positioning,fit,backgrounds}
\usetikzlibrary{calc}

\usetikzlibrary{arrows,positioning,fit,backgrounds,shapes}
\usetikzlibrary{calc}
\newtheorem{thm}{Theorem}

\theoremstyle{definition}
\newtheorem{defn}[thm]{Definition}
\theoremstyle{remark}
\newtheorem{rem}[thm]{Remark}

\begin{document}

\title{Secret Key Generation with One Communicator and a One-Shot Converse via Hypercontractivity}%
\author{\IEEEauthorblockN{Jingbo Liu~~~~~~~~Paul Cuff~~~~~~~~~Sergio Verd\'{u}}
\IEEEauthorblockA{Dept. of Electrical Eng., Princeton University, NJ 08544\\
\{jingbo,cuff,verdu\}@princeton.edu}}%
\maketitle

\begin{abstract}
A new model of multi-party secret key agreement is proposed, in which one terminal called the communicator can transmit public messages to other terminals before all terminals agree on a secret key. A single-letter characterization of the achievable region is derived in the stationary memoryless case. The new model generalizes some other (old and new) models of key agreement. In particular, key generation with an omniscient helper is the special case where the communicator knows all sources, for which we derive a zero-rate one-shot converse for the secret key per bit of communication.
\end{abstract}

\section{Introduction}
A random number known only to several geographically distributed terminals is a resource that can be used for cryptographic purposes such as secure communications. Remarkably, the terminals can usually distill such a shared random number, or secret key, by communicating information about certain correlated random processes they observe individually, even though the communication is wiretapped by some eavesdropper. The fundamental limits on the maximal secret key rate can be studied using information theoretic tools \cite{maurer1993secret}\cite{ahlswede1993common}\cite{csiszar2000common}.

In this paper we propose a new protocol of multi-party secret key agreement, called \emph{secret key generation with one communicator}, as shown in Figure~\ref{f_1com}. Terminals\footnote{Following the convention in \cite{ahlswede1993common}, we denote the terminals by the alphabets of the sources they observe.}
$\mathcal{Z},\mathcal{X}_1,\dots,\mathcal{X}_m$ observe general sources $Z,X_1,\dots,X_m$, respectively. The communicator $\mathcal{Z}$ is allowed to send public messages $W_1,\dots,W_m$ to $\mathcal{X}_1,\dots,\mathcal{X}_m$, before all the $m+1$ terminals agree on an integer $K$ (the key). We assume that for each $l\in\{1,\dots,m\}$ there is an eavesdropper wiretapping the communication link from the communicator to $\mathcal{X}_l$. Independence of $K$ and $W_l$ for each $l\in\{1,\dots,m\}$ ensures security.

We derive a single-letter characterization of the achievable public communication rates and the key rate in the stationary memoryless case, which is a special case of the above formulation where we identify $Z,X_1,\dots,X_m$ with the corresponding block symbols.
\begin{figure}[h!]
  \centering
\begin{tikzpicture}
[node distance=1cm,minimum height=10mm,minimum width=14mm,arw/.style={->,>=stealth'}]
  \node[rectangle,draw,rounded corners] (A) {$\mathcal{X}_1$};
  \node[rectangle,draw,rounded corners] (B) [right=of A] {$\mathcal{X}_2$};
  \node[rectangle] (C) [right =of B] {$\dots$};
  \node[rectangle,draw,rounded corners] (D) [right =of C] {$\mathcal{X}_m$};
  \node[rectangle,draw,rounded corners] (T) [above right=of B, xshift=-13mm, yshift=10mm] {Communicator ($\mathcal{Z}$)};
  \node[rectangle] (K1) [below =0.4cm of A] {$K_1$};
  \node[rectangle] (K2) [below =0.4cm of B] {$K_2$};
  \node[rectangle] (Km) [below =0.4cm of D] {$K_m$};
  \node[rectangle] (K) [above =0.4cm of T] {$K$};

  \draw [arw] (A) to node[midway,above]{} (K1);
  \draw [arw] (B) to node[midway,above]{} (K2);
  \draw [arw] (D) to node[midway,above]{} (Km);
  \draw [arw] (T) to node[]{} (K);
  \draw [arw,line width=1.5pt] (T) to node[midway,left]{$W_1$} (A.north);
  \draw [arw,line width=1.5pt] (T) to node[midway,left]{$W_2$} (B.north);
  \draw [arw,line width=1.5pt] (T) to node[midway,left]{$W_m$} (D.north);
\end{tikzpicture}
\caption{Key generation with one communicator}
\label{f_1com}
\end{figure}
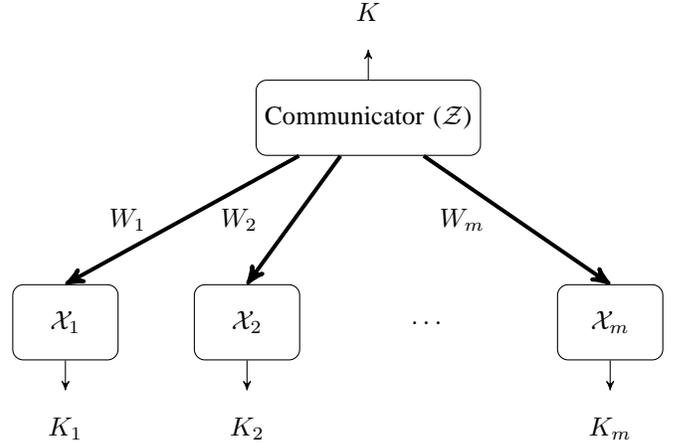

Of course, other related protocols of key generation have been studied in the literature. The canonical one-way protocol (Model~S with forward communication in \cite{ahlswede1993common}) is a special case of the secret key generation with one communicator protocol where $m=1$, in which case a key rate
\begin{align}\label{e1}
R<I({\sf U;X}_1)
\end{align}
is achievable with a communication rate
\begin{align}
R_1>I({\sf U;Z})-I({\sf U;X}_1)
\end{align}
where $Q_{{\sf ZX}_1}$ is the per-letter distribution of the stationary memoryless source and ${\sf U-Z-X}_1$. However the $m=1$ case does not assume the full complexity and difficulty of the general case, as we shall see later in terms of the single letter region and the coding scheme.

If no communication constraints are imposed, then the maximal key rate is (c.f.~\cite{csiszar2004secrecy}):
\begin{align}\label{e_3}
\min_{1\le l\le m} I({\sf Z;X}_l).
\end{align}
While random binning $Z^n$ shows the achievability of \eqref{e_3}, it cannot be used for the rate constraint case, where the receivers do not need to be able the construct $Z^n$; indeed a main difficulty with the rate constraint is to decide what common message should the terminals be able to agree on.

Finally if the terminals only need to construct a common random number without any secrecy guarantee (the CR generation problem), then the rate region is also known \cite[Theorem~4.2]{ahlswede1998common} (see also \cite{prabhakaran2014assisted}): a CR rate of
\begin{align}\label{e_cr1}
R<I({\sf U;Z})
\end{align}
is achievable if
\begin{align}\label{e_cr2}
R_l>I({\sf U;Z})-I({\sf U;X}_l),\quad\forall l=1,\dots,m
\end{align}
where ${\sf U-Z-X}^m$, which can be shown using an extension of source coding with side information \cite{slepian1973noiseless}. In contrast, the key generation problem in this paper requires much more involved achievability construction and analysis. Specifically, we use superposition coding in a novel way in order to convey the information of the key securely to the receivers. Whereas in the common usage of superposition coding the lower layer codeword is decoded before the upper layer \cite{cover1972broadcast}, in our construction the index of the upper layer codeword is transmitted to the receiver to facilitate the decoding of the lower layer codeword. Moreover, we use the recent achievability technique of likelihood encoding \cite{song} in order to simplify the security analysis considerably.

Particularly interesting is the special case of $Z=X^m$, which we call the \emph{omniscient helper} problem. In this case, a zero-rate one-shot converse on the secret key per bit of communication can be derived using hypercontractivity\footnote{Indeed, the new model stems from the first author's attempt to design a key agreement protocol in which the secret key per unit cost has a clean correspondence to hypercontractivity.}, strengthening the best converse bound that can be obtained from Fano's inequality. This new converse, derived from first principles, also underlines the intimate interplay between key agreement and hypercontractivity.

\section{Problem Setup and Main Results}
Let $Q_{ZX_1^m}$ be the joint distribution of sources $Z,X_1,\dots,X_m$.
As in Figure~\ref{f_1com}, the Terminals $\mathcal{Z},\mathcal{X}_1,\dots,\mathcal{X}_m$ observe $Z,X_1,\dots,X_m$, respectively, and the communicator $\mathcal{Z}$ computes the integers $W_1(Z),\dots,W_m(Z)$ possibly stochastically and sends them to $\mathcal{X}_1,\dots,\mathcal{X}_m$, respectively. Then, the $m+1$ parties calculate integers $K(Z),K_1(X_1,W_1),\dots,K_m(X_m,W_m)$ possibly stochastically.

In the case of stationary memoryless sources and block coding, we substitute $Z\leftarrow Z^n$ and $X_l\leftarrow {X_l}^n$ for each $l$, where $n$ is the blocklength.
The measures of reliable communication and secrecy are defined as follows:
\begin{align}
\epsilon_n&=\max_{1\le l\le m}\mathbb{P}[K\neq K_l],\label{e_m1}
\\
\nu_n&=\max_{1\le l\le m}\{\log|\mathcal{K}|-H(K|W_l)\}.\label{e_m2}
\end{align}
\begin{defn}
The $(m+1)$-tuple $(R,R_1,\dots,R_m)$ is said to be \emph{achievable} if a sequence of key generation schemes can be designed to fulfill the following conditions:
\begin{align}
\liminf_{n\to\infty}\frac{1}{n}\log|\mathcal{K}|&\ge R;
\\
\limsup_{n\to\infty}\frac{1}{n}\log|\mathcal{W}_l|&\le R_l,\quad l=1,\dots,m;
\\
\lim_{n\to\infty}\epsilon_n&=0;\\
\lim_{n\to\infty}\nu_n &=0.
\end{align}
\end{defn}
From the standard diagonalization argument~\cite{han_s}, the achievable region is closed. Our main result is the following:

\begin{thm}\label{thm_1}
The set of achievable rates is the closure of
\begin{align}
\bigcup_{Q_{{\sf US}^m|{\sf Z}}}
\left\{
\begin{array}{c}
  (R,R_1,\dots,R_m): \quad R\le \\
\min\{I({\sf U;Z}),I({\sf US}_1;{\sf X}_1),\dots,I({\sf US}_m;{\sf X}_m)\} \\
  R_l\ge I({\sf US}_l;{\sf Z|X}_l),\quad 1\le l\le m
\end{array}\right\}.
\label{reg0_1}
\end{align}
\end{thm}
\begin{rem}
The region in Theorem~\ref{thm_1} is not decreased if we restrict to the union over $Q_{{\sf U}|{\sf Z}}\prod_{l=1}^m Q_{{\sf S}_l|{\sf UZ}}$.
\end{rem}
\begin{rem}
Previous results of Ahlswede-Csisz\'{a}r \cite{ahlswede1993common} and Csisz\'{a}r-Narayan \cite{csiszar2004secrecy} shown in \eqref{e1}-\eqref{e_3} are clearly special cases of Theorem~\ref{thm_1}.
\end{rem}

\subsection{Special Case: Omniscient Helper}
As alluded to in the Introduction, the problem reduces to an interesting special case when the communicator knows all other sources. In this situation the communicator can be viewed as a helper, since the requirement that it can recover the key is vacuous because of its omniscience; the rate region in Theorem~\ref{thm_1} can also be simplified as follows, since setting ${\sf S}_l={\sf X}_l$ in \eqref{reg0_1} is optimal.
\begin{thm}\label{thm_2}
In the special case of ${\sf Z=X}^m$, the set of achievable rates is the closure of
\begin{align}\label{e_om}
\bigcup_{Q_{{\sf U|X}^m}}
\left\{
\begin{array}{c}
  (R,R_1,\dots,R_m):  \\
  R\le \min\{I({{\sf U;X}^m}),H({\sf X}_1),\dots,H({\sf X}_m)\}; \\
  R_l\ge I({\sf U};{\sf X}^m)-I({\sf U};{\sf X}_l),\quad 1\le l\le m
\end{array}\right\}.
\end{align}
\end{thm}
The region in \eqref{e_om} has some special features:
\begin{rem}\label{rem_41}
The region for a product source can be strictly larger than the Minkowski sum of its factors. Indeed even with unconstrained communication rates, the supremum key rate is as in \eqref{e_3}, where the minimum implies that a joint encoding can asymptotically strictly outperform separate encoding of the independent components.
\end{rem}
\begin{rem}\label{rem_42}
The key rate can be positive even if ${\sf X}^m$ has independent coordinates. For example, when $m=2$ and ${\sf X}_1\perp{\sf X}_2$ are equiprobable binary random variables, a key rate of $R=1$ is achievable if the helper sends ${X_1}^n\oplus{X_2}^n$ to $\mathcal{X}_2$ and, thus, the terminals agree on ${X_1}^n$.
\end{rem}
\begin{rem}\label{rem_43}
Comparing the CR generation region \eqref{e_cr1},\eqref{e_cr2} and the key generation region \eqref{e_om}, we see that in the omniscient helper problem the secrecy constraint does not increase the required communication rates as long as $R\le\min_{1\le l\le m}H({\sf X}_l)$. In particular, this is unconditionally true for continuous sources with infinite entropy. But even when the rate regions coincide, the underlying achievability constructions are different; indeed the coding schemes for CR generation in (\cite[Theorem~4.2]{ahlswede1998common} and \cite{prabhakaran2014assisted}) do not provide security. The reason why secrecy can be gained with no extra cost for small $R$ is that the helper shares sufficient secure randomness (the sources) with the other terminals to protect its messages.
\end{rem}

\section{One-Shot Achievability via Likelihood Encoder}
We outline the derivation of a one-shot achievability bound using the recent proof technique of likelihood encoding \cite{cuff2012distributed}\cite{song}. This method adapts to general non-discrete, non-i.i.d.~sources and simplifies the analysis of the secrecy constraint. Some standard notations in one-shot information theory, which may be found in reference \cite{verdu2012non}, will be used in this section.
\begin{thm}\label{thm_ach}
Suppose the sources have joint distribution $Q_{ZX^m}$. Fix an arbitrary $Q_{U|Z}$, $Q_{S_lU|Z}$, $1\le l\le m$, integers $I_0,\dots,I_m$ and $J_1\dots,J_m$. Then there exists a key generation scheme with $|\mathcal{K}|=I_0$, $|\mathcal{W}_l|=\prod_{i=1}^lI_iJ_l$ which guarantees that
\begin{align}
&\quad\mathbb{P}[K\neq K_l]\le 2m(\epsilon+T+T_l),
\\
&\quad\log |\mathcal{K}|-H(K|W_l)\nonumber
\\
&\le \inf_{0<\delta<(IJ_l)^{\frac{3}{2}}\exp(-1)}\left\{4m(2T+T_l+2\delta)\log\frac{(IJ_l)^{\frac{3}{2}}}{\delta}\right\}
\end{align}
for each $1\le l\le m$, where we have defined
\begin{align}
I:=\prod_{l=0}^m I_l,
\end{align}
\begin{align}
T:=\inf_{\gamma\in\mathbb{R}}\left\{\mathbb{P}[\imath_{U;Z}(U;Z)>\gamma]
+\frac{\exp(\frac{\gamma}{2})}{2\sqrt{I}}\right\},\label{e_t}
\end{align}
\begin{align}\label{e_tl}
T_l:=\inf_{\gamma\in\mathbb{R}}\left\{\mathbb{P}[\imath_{S_l;Z|U}(S_l;Z|U)>\gamma]
+\frac{\exp(\frac{\gamma}{2})}{2\sqrt{J_l}}\right\},
\end{align}
\begin{align}\label{e_ep}
\epsilon&:=\max_{1\le l\le m}\inf_{\gamma\in\mathbb{R}}\{
\mathbb{P}[\imath_{US_l;X_l}(US_l;X_l)\le\log(I_0-1)+\gamma]
\nonumber
\\
&\quad+\exp(-\gamma)\}.
\end{align}
\end{thm}
\begin{proof}[Proof Sketch]
\begin{itemize}
  \item \emph{Codebook construction}: for each $l=1,\dots,m$ define the set
      \begin{align}
      \mathcal{I}_l:=\{1,\dots,I_l\}.
      \end{align}
      Construct a codebook $u(i_0,i_1,\dots,i_m)$, $i_l\in\mathcal{I}_l$, $0\le l\le m$, where each codeword is generated i.i.d.~according to $Q_U$. Let $\mathcal{I}:=\mathcal{I}_0\times\mathcal{I}_1\times\dots\times\mathcal{I}_m$
      and $I=|\mathcal{I}|$.
      For each $i\in\mathcal{I}$ and $1\le l\le m$, independently generate a codebook
      \begin{align}
      [s_l(i,j)]_{j\in\mathcal{J}_l}
      \end{align}
      where each codeword is generated i.i.d.~according to $Q_{S_l|U=u(i)}$.

      \item \emph{Encoding}: define $\mu_V$ as the equiprobable distribution on $\mathcal{I}$, and
      \begin{align}
      \hat{P}_{Z|V=i}&:=Q_{Z|U=u(i)},\quad \forall i;
      \\
      \hat{P}_{ZV}&:=\hat{P}_{Z|V}\mu_V.
      \end{align}
      Moreover for each $l$, let $\mu_{\tilde{W}_l}$ be the equiprobable distribution on $\mathcal{J}_l$, and define
      \begin{align}
      P_{Z|\tilde{W}_l=jV=i}^{(l)}&:=Q_{Z|S_l=s_l(i,j)U=u(i)},\forall i,j;
      \\
      P_{Z\tilde{W}_lV}^{(l)}&:=
      P_{Z|\tilde{W}_lV}^{(l)}\mu_{\tilde{W}_l}\mu_{V}.
      \end{align}
      Then the encoder is a stochastic map
      \begin{align}\label{e_li2}
      \pi_{V\tilde{W}^m|Z}:=\hat{P}_{V|Z}\prod_{l=1}^mP_{\tilde{W}_l|VZ}^{(l)}
      \end{align}
      that maps the observation $z\in\mathcal{Z}$ to $v\in\mathcal{I}$ and $\tilde{w}^m\in\mathcal{J}^m$. In other words, we first find $v$ using a likelihood encoder with the likelihood function $\hat{P}_{Z=z|V}$ and then find $\tilde{w}_l$ using a likelihood encoder with the likelihood function $P_{Z=z|\tilde{W}_lV=v}^{(l)}$.
    Suppose $v=(v_0,v_1,\dots,v_m)$ where $v_l\in \mathcal{I}_l$, $0\le l\le m$. We identify $k=v_0$ and $w_l=(\tilde{w}_l,v_1^l)$ as the key for the communicator and the public messages. Note that the second components of $w_l$ have a nested (aligned) structure, which is important for maximizing the key rate.

    \item \emph{Error analysis}: The main idea is to use the soft covering lemma (c.f.~ \cite[Theorem~VII.1]{cuff2012distributed} or \cite{hayashi2006general}) iteratively to show that the true distribution $\pi_{\tilde{W}_lVZ}$ is close to $\hat{P}_{\tilde{W}_lVZ}^{(l)}$ in total variation (expected over the codebook). By construction $\tilde{W}_l$ and $V$ are independent under $\hat{P}^{(l)}$, implying that the individual message and the key are also nearly independent under $\pi$. Moreover, the decoding error probability of the receivers under $\hat{P}^{(l)}$ can be bounded directly by Shannon's achievability bound \cite{shannon1957certain}.
\end{itemize}
\end{proof}
Theorem~\ref{thm_ach} immediately implies the achievability part of the region \eqref{reg0_1} in the i.i.d.~case: assume without loss of generality that the sources are ordered in such a way that $I({\sf US}_l;{\sf X}_l)$
is non-increasing in $l$. We then identify $(S^{m},U,X_1,X_2,\dots,X_m,Z)$ in Theorem~\ref{thm_ach} as the block-coding counterpart $(S^{mn},U^n,{X_1}^n,\dots, {X_m}^n,Z^n)$ and let $\delta$ be exponentially converging to zero as $n\to\infty$, and
\begin{align}
J_l&:=\exp(n(I({\sf S}_l;{\sf X}_l|{\sf U})+\beta))),\quad l=1,\dots,m;
\\
I_0&:=\exp(n(\min\{I({\sf U;Z}),I({\sf US}_m;{\sf X}_m)\}-\beta));
\\
I_l&:=\exp(n(\min\{I({\sf U;Z}),I({\sf US}_{l-1};{\sf X}_{l-1})\}\nonumber
\\
    &\quad-\min\{I({\sf U;Z}),I({\sf US}_l;{\sf X}_l)\})),\quad l=2,\dots,m;
\\
I_1&:=\exp(n(I({\sf U;Z})-\min\{I({\sf U;Z}),I({\sf US}_1;{\sf X}_1)\}))
\end{align}
to show the achievability of rates
\begin{align}
R&:=\min\{I({\sf U;Z}),I({\sf US}_m;{\sf X}_m)\}-\beta;
\\
R_l&:=\max\{I({\sf S}_l;{\sf Z|U}),I({\sf US}_l;{\sf Z|X}_l)\}+3\beta,\quad 1\le l\le m
\end{align}
for $\beta>0$ arbitrary. This establishes the achievability of
\begin{align}
\bigcup_{Q_{{\sf US}^m|{\sf Z}}}
\left\{
\begin{array}{c}
  (R,R_1,\dots,R_m):
  \quad R\le
   \\
   \min\{I({\sf U;Z}),I({\sf US}_1;{\sf X}_1),\dots,I({\sf US}_m;{\sf X}_m)\} \\
  R_l\ge \max\{I({\sf S}_l;{\sf Z|U}),I({\sf US}_l;{\sf Z|X}_l)\},\quad \forall l
\end{array}\right\}.
\label{e_26}
\end{align}
Then the achievability of \eqref{reg0_1} follows by noting that the boundary of \eqref{e_26} can be achieved when ${\sf S}_l$ is chosen so that the two terms in the max are equal.

\section{Converse}
Due to space, this section only presents the main idea for the converse of Theorem~\ref{thm_1}.
\subsection{Deterministic Encoder}\label{sec_det}
We first consider the case where $K$ and $W^m$ are functions of $Z^n$ (but $\mathcal{X}_l$ are allowed to calculate their keys randomly from $(W_l,{X_l}^n)$, for $1\le l\le m$).
Given a key generation scheme, denote by $K,K_1,\dots,K_m$ the keys produced by $\mathcal{Z},\mathcal{X}_1,\dots\mathcal{X}_m$ and $W_1,W_2,\dots,W_m$ the messages sent to $\mathcal{X}_1,\dots,\mathcal{X}_m$. Define
\begin{align}
U_i&:=(K,Z^{i-1});\label{e_u}
\\
S_{li}&:=(W_l,{X_l}^{i-1}),\quad 1\le l\le m, 1\le i\le n\label{e_s}
\end{align}
and let $N$ be equiprobable on $\{1,\dots,n\}$ independent of all previously defined random variables. We identify
\begin{align}\label{e_us}
{\sf U}=U_N, \quad{\sf S}_l=S_{lN},\quad\forall l,
\end{align}
which fulfills that
\begin{align}
({\sf U},{\sf S}_1,\dots,{\sf S}_m)-{\sf Z}-({\sf X}_1,\dots,{\sf X}_m).
\end{align}
The bounds in Theorem~\ref{thm_1} can be verified using entropic manipulations and Fano's inequality.

\subsection{Stochastic Encoders}
The converse for stochastic encoders cannot be obtained by simple modifications of the analysis in \ref{sec_det}. Indeed, the bound in \eqref{reg0_1} no longer holds for stochastic encoders if we stick to the assignment of the auxiliary random variables in \eqref{e_u}-\eqref{e_us}. An alternative approach is to view a stochastic encoder as a deterministic function of $Z^n$ and $V$, where $V$ is a random number satisfying $({X_1}^n,\dots,{X_m}^n)-Z^n-V$, and then employ the converse for deterministic encoders. We immediately see that any achievable rates $(R,R_1,\dots,R_m)$ must satisfy
\begin{align}
R&\le \min\{I({\sf U;ZV}),I({\sf US}_1;{\sf X}_1),\dots,I({\sf US}_m;{\sf X}_m)\};\label{regv1}
\\
R_l&\ge \max\{I({\sf S}_l;{\sf ZV|U}),
I({\sf US}_l;{\sf ZV}|{\sf X}_l)\},\quad 1\le l\le m\label{regv2}
\end{align}
for some $P_{\sf USTV|Z}$. Then it is possible to show that the region specified by \eqref{regv1}-\eqref{regv2} is equivalent to the region specified in \eqref{reg0_1} upon optimization.

\section{A Zero-Rate One-Shot Converse}
In this section we derive a novel one-shot bound, using hypercontractivity, on the maximum ratio of the log alphabet sizes of the key and the messages such that the key can be successfully generated in the omniscient helper problem. Since this ratio is supremized as the key rate and the communication rates tend to zero, such a converse bound may also be called a \emph{zero-rate} converse. The bound is asymptotically tight in the case of abundant correlated sources but limited communication rates, and gives a \emph{strong converse} as it shows that the total variation between the true and the correct distributions tends to the maximal value under appropriate rate conditions. On the other hand, previous works have obtained one-shot converses using smooth Renyi entropy \cite{renner2005simple} or the meta-converse idea \cite{hayashi2014secret}\cite{tyagi2014bound}, for which the asymptotic tightness are achieved in the other extreme of limited correlated sources but unlimited communications.

An $m$-tuple of random variables $(X_1,\dots,X_m)$ is said to be $(p_1,\dots,p_m)$-hypercontractive for $p_l\in[1,\infty)$, $l=1,\dots,m$ if
\begin{align}\label{hyper1}
\mathbb{E}\left[\prod_{l=1}^m f_l(X_l)\right]\le \prod_{l=1}^m\|f_l(X_l)\|_{p_l}
\end{align}
for all bounded real-valued measurable functions $f_l$ defined on $\mathcal{X}_l$, $l=1,\dots,m$. In \cite{nair}, Nair showed that \eqref{hyper1} is equivalent to the following inequality\footnote{In \cite{nair} the equivalence is demonstrated for $m=2$, but the method therein can be easily extended to the $m>2$ case.}
\begin{align}\label{e_41}
I(U;X^m)\ge \sum_{l=1}^m \frac{1}{p_l}I(U;X_l)
\end{align}
being valid for all $P_{U|X^m}$. Thus from Theorem~\ref{thm_2} and \eqref{e_41}, key generation cannot be accomplished asymptotically if
\begin{align}\label{e_convf}
R<\sum_{l=1}^m \frac{1}{p_l}(R-R_l);
\end{align}
while if $r_1,\dots,r_m$ satisfies the property that $1\ge\sum_{l=1}^m \frac{1}{p_l}(1-r_l)$ for all $(p_1,\dots,p_m)$ such that $(X_1,\dots,X_m)$ is $(p_1,\dots,p_m)$-hypercontractive, then there exists $(R,R_1,R_2,\dots,R_m)$ achievable such that $\frac{R_l}{R}=r_l$ for each $l=1,\dots,m$.

We prove a zero-rate one-shot converse for the omniscient helper problem.
Consider the one-shot case.
Suppose the (possibly stochastic) encoder for the public messages is specified by $P_{W^m|X^m}$ and the (possibly stochastic) decoder for the key is given by $\prod_{l=1}^m P_{K_l|X_lW_l}$. Let $\mu_{K^m}$ be the \emph{correct} distribution under which $K_1=K_2=\dots=K_m$ is equiprobably distributed on $\mathcal{K}$. Clearly, a small total variation $|P_{K^m}-\mu_{K^m}|$ implies both uniformity of the key distribution and a small probability of key disagreement.
\begin{thm}\label{thm_sconv}
In the omniscient helper problem, if the source $X^m$ is $(p_1,\dots,p_m)$-hypercontractive,\footnote{In the i.i.d.~case this is equivalent to the per-letter source ${\sf X}^m$ being $(p_1,\dots,p_m)$-hypercontractive by the tensorization property \cite{nair}.} then
\begin{align}
\frac{1}{2}|P_{K^m}-\mu_{K^m}|
\ge1-\frac{1}{|\mathcal{K}|}
-\left[|\mathcal{K}|\prod_{l=1}^m\left(\frac{|\mathcal{W}_l|}{|\mathcal{K}|}\right)^\frac{1}{p_l}\right]^{\frac{1}{\sum p_l^{-1}}}
\end{align}
\end{thm}
\begin{rem}
Theorem~\ref{thm_sconv} only concerns the performance of CR generation, which will provide an obvious upper bound on the performance of key generation. For the omniscient helper problem, it turns out to be tight because the highest key-communication ratio is achieved with small rates (by convexity of the achievable region), in which regime the secrecy constraint does not require higher communication rates (Remark~\ref{rem_43}).
\end{rem}
\begin{rem}
Theorem~\ref{thm_sconv} yields a stronger converse on the achievable ratio of the the log alphabet sizes of the key and the messages than Theorem~\ref{thm_2}, because:
\begin{itemize}
  \item The converse from Theorem~\ref{thm_2} is vacuous when the rates are zero. In contrast, Theorem~\ref{thm_sconv} is still applicable when the log size of the key alphabet grows sub-linearly in the blocklength. In fact, as long as
     \begin{align}
    \log|\mathcal{K}|-\sum_{l=1}^m \frac{1}{p_l}(\log|\mathcal{K}|-\log|\mathcal{W}_l|)\to-\infty
    \end{align}
    which is weaker than \eqref{e_convf}, Theorem~\ref{thm_sconv} implies that $|P_{K^m}-\mu_{K^m}|$ converges to 2.
  \item Even if \eqref{e_convf} holds, the converse of Theorem~\ref{thm_2} relying on Fano's inequality does not guarantee that the error probability in \eqref{e_m1} tends to 1. Moreover Theorem~\ref{thm_2} uses relative entropy as the secrecy measure \eqref{e_m2} (stronger than total variation), amounting to a weaker converse.
\end{itemize}
\end{rem}
\begin{proof}
For any $k\in\mathcal{K}$,
\begin{align}
&\quad\mathbb{P}\left[\bigcap_{l=1}^m\{K_l=k\}\right]
\\
&=\int_{\mathcal{X}^m}\sum_{w^m}\prod_{l=1}^m P_{K_l=k|X_lW_l=w_l}P_{W^m|X^m}{\rm d}P_{X^m}
\\
&\le \int_{\mathcal{X}^m}\max_{w^m}\prod_{l=1}^m P_{K_l=k|X_lW_l=w_l}{\rm d}P_{X^m}
\\
&\le \int_{\mathcal{X}^m}\prod_{l=1}^m \max_{w_l} P_{K_l=k|X_lW_l=w_l}{\rm d}P_{X^m}
\\
&\le \prod_{l=1}^m\left[\int_{\mathcal{X}_l} (\max_{w_l} P_{K_l=k|X_lW_l=w_l})^{p_l}{\rm d}P_{X_l}\right]^{\frac{1}{p_l}}\label{e_hyp}
\\
&\le \prod_{l=1}^m\left[\int_{\mathcal{X}_l} \max_{w_l} P_{K_l=k|X_lW_l=w_l}{\rm d}P_{X_l}\right]^{\frac{1}{p_l}}\label{e_pl}
\\
&\le \prod_{l=1}^m\left[\sum_{w_l}\int_{\mathcal{X}_l}  P_{K_l=k|X_lW_l=w_l}{\rm d}P_{X_l}\right]^{\frac{1}{p_l}}\label{e_la}
\end{align}
where
\begin{itemize}
  \item \eqref{e_hyp} uses the definition of hypercontractivity;
  \item \eqref{e_pl} uses $p_l>1$ and $\max_{w_l} P_{K_l=k|X_lW_l=w_l}\le1$.
\end{itemize}
Raising both sides of \eqref{e_la} to the power of $\frac{1}{\sum_{i=1}^m p_i^{-1}}$, we obtain
\begin{align}\label{e_55}
&\quad\mathbb{P}\left[\bigcap_{l=1}^m\{K_l=k\}\right]^{\frac{1}{\sum_i p_i^{-1}}}
\nonumber
\\
&\le\prod_{l=1}^m\left[\sum_{w_l}\int_{\mathcal{X}_l}  P_{K_l=k|X_lW_l=w_l}{\rm d}P_{X_l}\right]^{\frac{p_l^{-1}}{\sum_i p_i^{-1}}}
\end{align}
But the function $t^m\mapsto \prod_{l=1}^mt_l^{\frac{p_l^{-1}}{\sum_i p_i^{-1}}}$ is a concave function on $[0,\infty)^m$, so by Jensen's inequality,
\begin{align}
&\quad\frac{1}{|\mathcal{K}|}\sum_{k=1}^{|\mathcal{K}|}\prod_{l=1}^m\left[\sum_{w_l}\int_{\mathcal{X}_l}  P_{K_l=k|X_lW_l=w_l}{\rm d}P_{X_l}\right]^{\frac{p_l^{-1}}{\sum_i p_i^{-1}}}
\\
&\le
\prod_{l=1}^m\left[\sum_{w_l}\int_{\mathcal{X}_l}  \frac{1}{|\mathcal{K}|}\sum_{k=1}^{|\mathcal{K}|}P_{K_l=k|X_lW_l=w_l}{\rm d}P_{X_l}\right]^{\frac{p_l^{-1}}{\sum_i p_i^{-1}}}
\\
&=
\prod_{l=1}^m\left[\sum_{w_l}\int_{\mathcal{X}_l}  \frac{1}{|\mathcal{K}|}{\rm d}P_{X_l}\right]^{\frac{p_l^{-1}}{\sum_i p_i^{-1}}}
\\
&=
\prod_{l=1}^m\left(\frac{|\mathcal{W}_l|}{|\mathcal{K}|}\right)
^{\frac{p_l^{-1}}{\sum_i p_i^{-1}}}\label{e_59}
\end{align}
Combining \eqref{e_55} and \eqref{e_59} we obtain
\begin{align}\label{e_f0}
\frac{1}{|\mathcal{K}|}\sum_{k=1}^{|\mathcal{K}|}\mathbb{P}
\left[\bigcap_{l=1}^m\{K_l=k\}\right]^{\frac{1}{\sum_i p_i^{-1}}}
\le \prod_{l=1}^m\left(\frac{|\mathcal{W}_l|}
{|\mathcal{K}|}\right)^{\frac{p_l^{-1}}{\sum_i p_i^{-1}}}.
\end{align}
Finally we invoke the following elementary bound:
\begin{align}\label{e_f2}
&\quad\frac{1}{2}|P_{K^m}-\mu_{K^m}|
\nonumber
\\
&=\sum_{k=1}^{|\mathcal{K}|}\left|\mathbb{P}\left[\bigcap_{l=1}^m\{K_l=k\}\right]-\frac{1}{|\mathcal{K}|}\right|
+1-\sum_{k=1}^{|\mathcal{K}|}\mathbb{P}\left[\bigcap_{l=1}^m\{K_l=k\}\right]
\\
&\ge 1-\frac{1}{|\mathcal{K}|}-|\mathcal{K}|^{\frac{1}{\sum_l p_l^{-1}}-1}\sum_{k=1}^{|\mathcal{K}|}\mathbb{P}\left[\bigcap_{l=1}^m\{K_l=k\}\right]^{\frac{1}{\sum_l p_l^{-1}}}
\end{align}
and the proof is finished by combining \eqref{e_f0} and \eqref{e_f2}.
\end{proof}

\section{Discussion}
It remains an enticing problem for future research to find out whether the achievable region is changed if we further require that the key has to be independent of all messages, instead of each message individually (see \eqref{e_m2}). Such a stronger secrecy constraint is relevant when a powerful eavesdropper is able to intercept the messages to all the receivers. Our achievability proof does not guarantee this stronger level of secrecy, but for some specific sources it is possible to use structured codes to align different sub-codebooks so that the achievable rates do not change. Furthermore, in the unlimited communication case the key rate is not compromised by the stronger requirement either; see \eqref{e_3}. Generally, an inner bound can be obtained by replacing $S_l$ in \eqref{reg0_1} with $S^l$, the proof of which involves a vertical structure of superposition codebooks for $U$, $S_1$, ..., $S_m$, in contrast to the parallel structure of $S_1$, ..., $S_m$ in the achievability proof of Theorem~\ref{thm_1}.

Theorem~\ref{thm_sconv} also gives an asymptotically tight strong converse bound for the canonical one-way protocol (Model~S with forward communication in \cite{ahlswede1993common}). By setting $m=2$ and $R_1=0$, the resulting model immediately gives an upper-bound on the performance of a CR generation model where $\mathcal{X}_1$ communicates to $\mathcal{X}_2$. This in turn bounds performance of one-way key generation model because the public communication can be used as part of the CR. In the end we can show that the TV between $P_{KW}$ and the correct distribution $\mu_{KW}=\mu_K\mu_W$ tends to $2$ if $\log|\mathcal{K}|-\frac{s^*(X_1;X_2)}{1-s^*(X_1;X_2)}\log|\mathcal{W}|\to\infty$, where $s^*(X_1;X_2)$ is the strong data processing coefficient \cite{nair}.

\section{Acknowledgments}
The authors would like to thank Himanshu Tyagi, Shun Watanabe and Sudeep Kamath for helpful literature background information.
This work was supported by NSF under Grants CCF-1350595, CCF-1116013, CCF-1319299, CCF-1319304, and the Air Force Office of Scientific Research under Grant FA9550-12-1-0196.

\bibliographystyle{ieeetr}
\bibliography{ref_om}
\end{document}